\documentclass[leqno,sumlimits,intlimits,namelimits]{amsart}

\usepackage{bm}

\usepackage{amssymb,latexsym}
\usepackage[numeric,initials,bibtex-style,nobysame]{amsrefs}
\usepackage{amsthm}

\usepackage{a4wide}
\usepackage[american]{babel}

\usepackage{eucal}
\usepackage{mathrsfs}

\usepackage[dvipsnames]{xcolor}

\def\pg{\mathhexbox278}

\newcommand{\D}{\ensuremath{\mathcal{D}}}
\newcommand{\G}{\ensuremath{\mathcal{G}}}
\renewcommand{\S}{\mathscr{S}}
\newcommand{\E}{\ensuremath{\mathcal{E}}}

\newcommand{\mb}[1]{\ensuremath{\mathbb{#1}}}
\newcommand{\N}{\mb{N}}

\newcommand{\R}{\mb{R}}
\newcommand{\C}{\mb{C}}
\newcommand{\Z}{\mb{Z}}
\newcommand{\T}{\mb{T}}

\newcommand{\W}{\mathcal{W}}

\renewcommand{\d}{\ensuremath{\partial}}

\newfont{\bl}{msbm10 scaled \magstep2}

\newtheorem{theorem}{Theorem}[section]
\newtheorem{lemma}[theorem]{Lemma}
\newtheorem{proposition}[theorem]{Proposition}
\newtheorem{definition}[theorem]{Definition}
\newtheorem{corollary}[theorem]{Corollary}
\theoremstyle{definition}
\newtheorem{remark}[theorem]{Remark}
\newtheorem{example}[theorem]{Example}

\newcommand{\beq}{\begin{equation}}
\newcommand{\eeq}{\end{equation}}

\newcommand{\isom}{\cong}
\newcommand{\col}{\colon}

\newcommand{\FT}[1]{\widehat{#1}}
\newcommand{\F}{\ensuremath{{\mathcal F}}}

\newcommand{\dis}[2]{\langle #1 , #2 \rangle}
\newcommand{\inp}[2]{\langle #1 | #2 \rangle}  
\newcommand{\notmid}{\mid\kern-0.5em\not\kern0.5em}

\newcommand{\norm}[2]{{\left\| #1 \right\|}_{#2}}

\newcommand{\Norm}[1]{\norm{#1}{}}

\newcommand{\be}{\beta}

\newcommand{\de}{\delta}
\newcommand{\eps}{\varepsilon}

\newcommand{\vphi}{\varphi}

\newcommand{\la}{\lambda}

\newcommand{\om}{\omega}
\newcommand{\Om}{\Omega}
\newcommand{\sig}{\sigma}

\newcommand{\supp}{\mathop{\mathrm{supp}}}

\renewcommand{\Re}{\ensuremath{\mathop{\mathrm{Re}}}}
\renewcommand{\Im}{\ensuremath{\mathop{\mathrm{Im}}}}

\newcommand{\ovl}[1]{\overline{#1}}

\begin{document}

\pagestyle{plain}

\title{Dirac and normal states on Weyl--von Neumann algebras}

\author{G\"unther H\"ormann}

\address{Fakult\"at f\"ur Mathematik\\
Universit\"at Wien, Austria}

\email{guenther.hoermann@univie.ac.at}

\subjclass[2010]{Primary: 81R10; Secondary: 46F99}

\keywords{Weyl algebra, quantization with constraints, von Neumann algebras, generalized functions}

\date{\today}

\begin{abstract}
We study particular classes of states on the Weyl algebra $\mathcal{W}$ associated with a symplectic vector space $S$ and on the von Neumann algebras generated in representations of $\W$. Applications in quantum physics require an implementation of constraint equations, e.g., due to gauge conditions, and can be based on so-called Dirac states. The states can be characterized by nonlinear functions on $S$ and it turns out that those corresponding to non-trivial Dirac states are typically discontinuous. We discuss general aspects of this interplay between functions on $S$ and states, but also develop an analysis for a particular example class of non-trivial Dirac states. In the last part, we focus on the specific situation with $S = L^2(\mathbb{R}^n)$ or test functions on $\R^n$ and relate properties of states on $\W$ with those of generalized functions on $\mathbb{R}^n$ or with harmonic analysis aspects of corresponding Borel measures on Schwartz functions and on temperate distributions.
\end{abstract}

\maketitle


\section{Introduction and review of the $C^*$-algebraic foundation}

Grundling and Hurst (cf.\ \cite{GH:85} and a review in \cite{Grundling:06}) have developed a program to implement constraints in the $C^*$-algebraic framework for quantum systems. Such constraints are often arising from gauge conditions in field theories and have been discussed already in a sketchy form by Dirac in the 1950's. The application to the \emph{$C^*$-algebra of the canonical commutation relations} or \emph{Weyl algebra} is of particular importance regarding the interplay with so-called regular representations (\cite{GH:88}) that guarantee the existence of corresponding unbounded field operators. In this context, attempts at rigorous $C^*$-algebraic constructions of basic features of quantum electrodynamics have been started in \cite{Grundling:88}, with improvements supplied by Narnhofer and Thirring in \cite{NT:92}.

In Subsection 1.1 we recall the definition of the Weyl or CCR algebra associated with a symplectic space $(S,\be)$, while Subsection 1.2 provides a review of Dirac states. Section 2 discusses these notions in the context of representations and normal states on the corresponding von Neumann algebras. Section 3 focusses on the case $S = L^2(\R^n)$ and puts states on the Weyl algebra into a context with generalized functions on $\R^n$ and measures on Schwartz functions and temperate distributions.

\subsection{Definition of the Weyl algebra and basic notions. }

Let $S$ be a real vector space with a nondegenerate symplectic form $\be$, i.e., a bilinear skew-symmetric map $\be \col S \times S \to \R$ such that $\be(y,z) = 0$ for all $y \in S$ implies $z = 0$.
Our main example will be to start with a real Hilbert space $(Q,(.|.))$ and then equip $S := Q \times Q$ with the nondegenerate symplectic form $\be$, given by
\beq\label{MainBeta}
    \be(y,z) := \frac{1}{2} ( (y_1|z_2) - (y_2|z_1) ) \quad \forall y = (y_1, y_2), z = (z_1, z_2) \in S.
\eeq
We can then consider $S$ as a complex Hilbert space: Define the multiplication  of $(z_1,z_2) \in Q \times Q$ by a complex scalar $r + i s$ simply as $(r + is) \cdot (z_1, z_2) := (r z_1 - s z_2, s z_1 + r z_2)$ (this is isomorphic to the complexification via the real tensor product $\C \otimes Q$). The complex inner product is then defined by
$$
     (y,z)_\C := \be(y,iz) + i \be(y,z) \quad \forall y, z \in S. 
$$
\begin{example}\label{Specases} The above construction includes two important special cases:

\noindent (i) $Q = \R^n$ with the standard inner product, which gives $S \isom \R^{2n}$ as the standard symplectic vector space; alternatively, the latter is described as $\C^n$ with the (real) symplectic form $\be(y,z) = \Im(\ovl{y}^T \cdot z )/2$. It is the reference model for quantum systems with finitely many degrees of freedom.

\noindent (ii) We obtain the complex Hilbert space $S \isom L^2(\R^n)$ from $Q$ being the subspace of real-valued functions and identifying $(f_1,f_2) \in Q \times Q$ with $f_1 + i f_2 \in L^2(\R^n)$. The (real) symplectic form is then simply given by $\be(f,g) = \Im \inp{f}{g}/2$, where $\inp{.}{.}$ denotes the standard complex $L^2$-inner product. This is typically the one-particle space for a Fock space model in quantum field theory.
\end{example}

The \emph{Weyl algebra over the symplectic space $(S,\be)$} will be denoted here by $\W(S,\be)$ (in the literature also $\text{CCR}(S,\be)$) or simply by $\W$ and is defined as the unique $C^*$-algebra (cf.\ \cite[Theorem 2.1]{Petz:90})  generated by a set $\{ W(z) \mid z \in S\}$ such that 
\beq\label{Weyl}
    \forall y, z \in S: \quad W(-z) = W(z)^* \quad\text{and}\quad W(y) W(z) = e^{i \be(y,z)} W(y + z). 
\eeq
In particular, every $W(z)$ is unitary and $W(0)$ is the unit in $\W$, which we simply denote by $1$.

Representations of the Weyl algebra as operator algebras on Hilbert spaces are at the heart of quantum physics. In particular, this is true of those representations that are \emph{regular} in the sense that for any $z_0 \in S$ there is a corresponding \emph{observable}, i.e., a self-adjoint generator of (the image of) the unitary group $(W(t z_0))_{t \in \R}$. In the case of finitely many degrees of freedom corresponding to $Q=\R^n$ as in (i) of the above example, the well-known von Neumann uniqueness result shows that, up to unitary equivalence, the Schr\"odinger representation is the unique irreducible regular representation, (cf.\  \cite[Corollary 5.2.15 and example 5.2.16]{BR:V2} or \cite[Proposition 1.1 and Theorem 1.2]{Petz:90}). Quantum field theory relies on infinitely many degrees of freedom, i.e., $S$ being not finite dimensional and it turns out that there exist uncountably many inequivalent irreducible regular representations.  An original proposal for a classification was outlined in \cite{GW:54} and information on further developments, with $S$ essentially as in Example \ref{Specases}(ii), can be found in \cite[Theorem 5.2.14 and Notes and Remarks Chapter 5]{BR:V2}. The case where $S$ is a locally convex vector space is discussed in  \cite{Hoermann:97} and \cite{Hoermann:93}, where also generalized Schr\"odinger and Fock representations are constructed.

\begin{remark} \label{Wtheta}
More generally, as was shown by Slawny in \cite{Slawny:72}, one can construct a unique Weyl algebra $\W(G,b)$ over any commutative group $G$ with a non-degenerate bicharacter $b \col G \times G \to S^1$, i.e., $b$ is a character in each argument separately and $b(x,y) b(x,y)^{-1} = 1$ for all $y \in G$ implies that $x = 0$. One can alternatively describe $\W(G,b)$ as a $C^*$-subalgebra of the group $C^*$-algebra for a corresponding Heisenberg group $G \times S^1$ as given in \cite[Proposition 1]{Hoermann:97}. In particular, this allows for a mathematical framework of certain non-commutative versions of the commutative $C^*$-algebra $C(S^1 \times S^1)$ of complex-valued continuous functions on the two-torus $S^1 \times S^1$. Note that the latter is generated by the countable family of functions $W(n)$ ($n \in \Z^2$), where $W(n)(z_1,z_2) := \exp(2 \pi i (n_1 z_1 + n_2 z_2))$ and clearly $W(m) W(n) = W(m+n)$ holds for all $m, n$ belonging to the  commutative group $\Z^2$. Some models studied in quantum statistical mechanics made use of a 
 $\theta$-deformed Weyl algebra $\mathcal{A}_\theta$ on the two-torus with real deformation parameter $\theta$ (cf., e.g.,  \cite{BNS:91,BNS:91E}). It can be described in terms of the group $\Z^2$ with $\theta$-dependent bicharacter, or more directly by the Weyl-type relations
$$ 
    W(m) W(n) = e^{i \theta (m_1 n_2 - m_2 n_1)} W(m + n) \quad (m,n \in \Z^2).
$$
The $C^*$-algebra $\mathcal{A}_\theta$ appears also in the context of rigorous models for the quantum Hall effect, where $\theta$ is proportional to the product of the magnetic field with the electric charge (cf.\ \cite[Section 8, Equation (4)]{Bellissard:86}). 
\end{remark}

\subsection{Review of Dirac states on $\W$}

Recall that a state over a $C^*$-algebra is a normalized positive linear functional. If $L$ is a subspace of $S$, e.g., representing constraints, then the Weyl relations \eqref{Weyl} show that $\W_L := \text{span}\{ W(y) \mid y \in L\}$ is a ${}^*$-subalgebra of $\W$, and the basic idea in \cite{GH:85} is that the physical states should have trivial values on the unitary generators $W(y)$.  
The following definition is slightly rephrasing the original notion from \cite[Equation (2.4)]{GH:85} and agrees with the variant used in \cite[Section 3]{GH:88}.

\begin{definition}  A state $\om$ over  $\W$ is called a \emph{Dirac state adapted to the subspace $L \subseteq S$}, if
$$
    \forall  y \in L \col \quad \om(W(y)) = 1.
$$
We will use the notion of a \emph{Dirac state on $\mathcal{A}$} for any state $\mu$ defined on a $C^*$-subalgebra $\mathcal{A} \subseteq \W$ with $\mathcal{A} \supseteq \{ W(y) \mid y \in L\}$ and such that $\mu(W(y)) = 1$ for every $y \in L$.
\end{definition}

\begin{example}\label{PGauge}
In the case $S = Q \times Q$ with a real Hilbert space and $\be$ as in \eqref{MainBeta}, we may consider the Lagrangian subspace 
 $L := \{0\} \times Q$ of ``momentum variables''. A Dirac state adapted to $L$ is then required to satisfy $\om(W(0,z_2)) = 1$ for every $z_2 \in Q$. We will discuss a construction of such a Dirac state below in Example \ref{PGauge2}.
\end{example}

According to \cite[Lemma 6.1]{GH:88CMP} (see also \cite[Theorem 17(i)]{Grundling:06}) we have that Dirac states adapted to $L$ exist  for the Weyl algebra if and only if the subspace $L$ is \emph{$\be$-isotropic}, i.e., 
\beq\label{Lisotropic}
   L \subseteq \{ z \in S \mid \forall y \in L \col \be(z,y) = 0 \} =: L^\be
\eeq
and we will always suppose this from now on. 

The following statement shows that the constraints are implemented in the GNS representation of a Dirac state by trivial action of the corresponding generators on the cyclic (vacuum) vector. In \cite[Definition 2.1]{Grundling:06}, the same property is even taken as the definition of Dirac states in case of unitary constraints. The result is included in the more general statement in \cite[Theorem 2.19(ii)]{GH:85}, but we will nevertheless give a simple direct proof in case of the Weyl algebra here.

\begin{proposition}\label{OmProp} Let $\om$ be a state over $\W$ and $\pi_\om \col \W \to B(H_\om)$ be the corresponding GNS representation with cyclic vector $\Om_\om$. Let $L$ be a subspace of $S$, then $\om$ is a Dirac state adapted to $L$, if and only if
$$
   \pi_\om(W(y)) \Om_\om = \Om_\om \quad  \forall y \in L.
$$
\end{proposition}
\begin{proof}  Recall that $H_\om$ is constructed as a completion of the quotient $\W / N_\om$ with the left ideal $N_\om = \{ A \in \W \mid \om(A^* A) = 0\}$,  the cyclic vector is $\Om_\om = 1 + N_\om$, and the operator action is given in general by $\pi_\om(A)(B + N_\om) := AB + N_\om$. Thus, $\pi_\om(W(y)) \Om_\om = W(y) + N_\om$ and it suffices to show that for every $y \in L$, $T_y:= W(y) - 1 \in N_\om$ if and only if $\om(W(y)) = 1$.  Thanks to $W(y)^* = W(-y) = W(y)^{-1}$ and the relations $\om(A^*) = \ovl{\om(A)}$  for arbitrary $A \in \W$, we have 
$$
    \om(T_y^* T_y) = \om((W(-y) -1)(W(y) - 1)) = \om(1 - W(-y) - W(y) + 1) = 2 - 2 \Re \om(W(y)).
$$
The condition $\om(T_y^* T_y) = 0$ is thus equivalent to $\Re \om(W(y)) = 1$. Since $|\om(W(y))| \leq 1$ due to unitarity of $W(y)$, we have, in fact, that $\om(T_y^* T_y) = 0$ is equivalent to $\om(W(y)) = 1$.
\end{proof}

Following \cite[Section 2]{GH:85}, one can improve the above result considerably on a $C^*$-subalgebra $\mathcal{O}$ of $\W$, which is used in implementing the \emph{algebra of observables compatible with the constraint conditions}. The construction of $\mathcal{O}$ is as follows: Let $\mathcal{L}$ be the $C^*$-subalgebra of $\W$ generated by $\{W(y) - 1 \mid y \in L\}$, $\mathcal{D}$ be the closure of $(\W \cdot \mathcal{L}) \cap  (\W \cdot \mathcal{L})^*$ in $\W$, and define 
\beq\label{DefO}
\mathcal{O} := \{ W \in \W \mid \forall D \in \mathcal{D}\col W D - D W =: [W,D] \in \mathcal{D}\}.  
\eeq
Both $\mathcal{D}$  and $\mathcal{O}$ are $C^*$-subalgebras and we clearly have $\mathcal{L} \subseteq \mathcal{D} \subseteq \mathcal{O}$. Note that $\mathcal{D}$ is a closed two-sided ideal in $\mathcal{O}$. We claim that $W(y) \in \mathcal{O}$ for every $y \in L$. In fact, it suffices to show that $[W(y),D] \in \mathcal{D}$ for every $D \in (\W \cdot \mathcal{L}) \cap  (\mathcal{L} \cdot \W)$, which is obvious upon writing $W(y)D - D W(y) = (W(y) -1) D - D (W(y) -1)$, since both terms of this last expression belong to $(\W \cdot \mathcal{L}) \cap  (\mathcal{L} \cdot \W) \subseteq \mathcal{D}$. The following result is from \cite[Theorem 2.20(iii)]{GH:85}.

\begin{proposition}\label{OProp} Let $L$ be a subspace of $S$ and $\mu$ be a Dirac state on $\mathcal{O}$ adapted to $L$. Then we have $\mathcal{D} \subseteq \ker \pi_\mu$, where $\pi_\mu$ denotes the GNS representation of $\mathcal{O}$ associated with $\mu$.
\end{proposition}

In particular, we obtain from the above proposition for a Dirac state $\mu$ on $\mathcal O$ and $y \in L$ that $\pi_\mu(W(y))$ is the identity, since $W(y) - 1 \in \mathcal{L} \subseteq \mathcal{D}$.  Although one can always extend $\mu$ to a (Dirac) state $\widetilde{\mu}$ on $\W$, we will not necessarily obtain $\mathcal{D} \subseteq \ker \pi_{\widetilde{\mu}}$. (We have $\mathcal{O} \cap \ker \pi_{\widetilde{\mu}} \subseteq \ker \pi_\mu$, but cannot expect equality in general. With $\widetilde{\mu} = \om_0$ as in Example \ref{PGauge2} and $\mu$ its restriction to $\mathcal{O}$ we obtain a counterexample, because $W(0,z_2) - 1$ with $z_2 \neq 0$ is not mapped to the zero operator in the GNS representation $\pi_{\om_0}$.) If we consider instead a situation where we start with a Dirac state $\om$ on $\W$, then we may at least conclude that the operators $\pi_\om(W(y))$ with $y \in L$ act trivially on a certain subspace of $H_\om$ generated from $\mathcal{O}$.
 
\begin{corollary} Let $\om$ be a Dirac state over $\W$ adapted to the subspace $L$ of $S$. Let $\pi_\om \col \W \to B(H_\om)$ denote the corresponding GNS representation and define the closed subspace $H_\om(\mathcal{O}) \subseteq H_\om$ as the closure of $\text{\rm span}\{  \pi_\om(A) \Om_\om \mid A \in \mathcal{O}\}$. If $ y \in L$, then  $\pi_\om(W(y)) \xi = \xi$ for every $\xi \in H_\om(\mathcal{O})$. 
\end{corollary}
\begin{proof} Let $\mu$ denote the restriction of $\om$ to $\mathcal{O}$,  which gives a positive linear functional on $\mathcal{O}$. Since $1 \in \mathcal{O}$ and $\mu(1) = \om(1) = 1$, $\mu$ is a state on $\mathcal{O}$. It is a Dirac state on $\mathcal{O}$, since $\mu(W(y)) = \om(W(y)) = 1$ for $y \in L$. By Proposition \ref{OProp}, we have $\mathcal{D} \subseteq \ker \pi_\mu$, hence $0 = \inp{\pi_\mu(B^*) \Om_\mu}{\pi_\mu(D) \pi_\mu(C) \Om_\mu} = \inp{\Om_\mu}{\pi_\mu(B D C) \Om_\mu} = \mu(B D C) = \om(B D C)$ holds for every $D \in \mathcal{D}$ and for all $B, C \in \mathcal{O}$. Since $\om(BDC) = \inp{\pi_\om(B^*)\Om_\om}{\pi_\om(D) \pi_\om(C)\Om_\om}$, we obtain that $\pi_\om(D) \xi = 0$ for $D \in \mathcal{D}$ and $\xi \in \text{span} \{ \pi_\om(A) \Om_\om \mid A \in \mathcal{O}\} = H_\om(\mathcal{O})$, which proves the claim upon setting $D = W(y) -1$.
\end{proof}

Recall that a state $\om$ on the Weyl algebra is called \emph{regular}, if the corresponding GNS representation $\pi_\om \col \W \to B(H_\om)$ is \emph{regular}, that is, for every $z \in S$, the map $\R \to B(H_\om)$ given by $t \mapsto \pi_\om(W(t z))$ is continuous with respect to the strong operator topology on $B(H_\om)$, i.e., $t \mapsto \pi_\om(W(t z)) \xi$ is continuous $\R \to H_\om$ for every $\xi \in H_\om$. The relevance of this notion for physics stems from the fact that it guarantees the existence of self-adjoint field operators $\Phi(z)$ ($z \in S$) as generators of the unitary groups $(\pi_\om(W(tz)))_{t \in \R}$. 

\begin{remark} We recall that for any $^*$-representation $\pi \col \W \to B(H)$ of $\W$ on some Hilbert space $H$, one has equivalence of regularity in the sense of the strong operator topology with that in the weak operator topology. Trivially, the former implies the latter, which is characterized by requiring that for every $z \in S$, the map $t \mapsto \inp{\xi}{\pi(W(tz))\eta}$ is continuous $\R \to \C$ for arbitrary $\xi, \eta \in H$. To prove the reverse implication, we first note that by the group property of $(\pi(W(tz)))_{t \in \R}$, it suffices to show continuity at $t = 0$, and then apply the Weyl relations to observe
\begin{multline*}
    \Norm{\pi(W(tz)) \xi - \xi}^2 = \inp{\xi}{\pi(W(tz) -1)^* \pi(W(tz) - 1) \xi} =
    \inp{\xi}{\pi((W(tz) - 1)^*(W(tz) -1)) \xi}\\  = \inp{\xi}{\pi((W(-tz) - 1)(W(tz) -1)) \xi} = 
    2 \Norm{\xi}^2 - \inp{\xi}{\pi(W(-tz))\xi} - \inp{\xi}{\pi(W(tz))\xi}.
\end{multline*}
\end{remark}

If $\om$ is a regular state on $\W$, then we clearly have  for every $z \in S$ that the map $f_z \col \R \to \C$ with $f_z(t) := \inp{\Om_\om}{\pi_\om(W(t z)) \Om_\om} = \om(W(tz))$ is continuous, in particular, $\lim_{t \to 0} f_z(t) = f_z(0) = \om(1) = 1$. According to \cite[Proposition 3.5]{Petz:90} even a converse of this is true, thus we have the following result.
\begin{proposition}\label{RegProp} A state $\om$ over $\W$ is regular, if and only if $\lim\limits_{t \to 0} \om(W(t z)) = 1$ for every $z \in S$. 
\end{proposition}

The condition for a Dirac state can be put into the following alternative form, which follows from a discussion around Theorem 2.6(ii) in  \cite{GH:85}, but we give a simple adaptation of its proof here.
\begin{lemma}\label{altformlem} A state $\om$ over $\W$ is a Dirac state adapted to $L$, if and only if we have 
$$
    \forall A \in \W, \forall y \in L \col \quad \om(A W(y)) = \om(A) = \om(W(y) A).
$$
\end{lemma}
\begin{proof} Clearly, if $\om$ satisfies the above, then putting $A=1$ shows that it is a Dirac state. For the converse statement, suppose $\om$ is a Dirac state adapted to $L$, $y \in L$, and $A \in \W$. As noted already in the proof of Proposition \ref{OmProp}, we have for $T_y := W(y) - 1$ that $\om(T_y^* T_y) = 0$. Therefore the Cauchy-Schwarz inequality yields 
$|\om(A T_y)| ^2 \leq \om(A A^*) \om(T_y^* T_y) = 0$, which implies $A T_y \in \ker \om$, i.e., $\om(A W(y)) = \om(A)$ holds for arbitrary $A \in \W$ and $y \in L$. Since $\ker \om$ is an involutive subset of $\W$ we obtain that also $W(y) A \in \ker \om$ for all $A \in \W$ and $y \in L$ and the proof is complete.
\end{proof}

It turns out that the requirement to be a Dirac state adapted to an isotropic subspace $L$ of $S$ is in conflict with regularity for that state. This result is shown  \cite[Theorem 3.1]{GH:88}, but it is instructive to repeat (and slightly simplify) its proof here.
\begin{theorem}\label{ThmNonreg} If $L \neq \{ 0\}$, then no Dirac state over $\W$ adapted to $L$ can be regular. 
\end{theorem}
\begin{proof} Let $\om$ be a Dirac state over $\W$ adapted to $L$. By assumption, we have $S \setminus L^\be \neq \emptyset$.\\[1mm] 
\emph{Claim:} Let $z_0 \in S \setminus L^\be$ be arbitrary, then $\om(W(z_0)) = 0$.\\[1mm] 
To prove the claim we choose $y_0 \in L$ such that $\be(y_0,z_0) \neq 0$.  By Lemma \ref{altformlem}, we have for every $t \in \R$,  $\om(W(z_0) W(t y_0)) = \om(W(z_0)) = \om(W(t y_0) W(z_0))$. Applying the Weyl relations  $W(z_0) W(t y_0) = \exp(- t i \be(y_0,z_0)) W(z_0 + t y_0)$ and $W(t y_0) W(z_0) = \exp( t i \be(y_0,z_0)) W(z_0 + t y_0)$ we obtain
$$
   \forall t \in \R \col \quad e^{- t i \be(y_0,z_0)} \om(W(z_0 + t y_0)) = \om(W(z_0)) = e^{ t i \be(y_0,z_0)} \om(W(z_0 + t y_0)).
$$
Since $\be(y_0,z_0) \neq 0$, the outermost members imply $\om(W(z_0 + t y_0)) = 0$ and thus $\om(W(z_0)) = 0$. 

If $z_0 \in S \setminus L^\be$, then $t z_0 \in S \setminus L^\be$ for every $t \neq 0$ and the claim applies then to $t z_0$ and shows that $\om(W(t z_0)) = 0$. We obtain that the map $f \col \R \to \C$, $t \mapsto \om(W(t z_0))$ is discontinuous at $t = 0$, because $f(0) = \om(W(0)) = \om(1) = 1$, whereas $f(t) = 0$ for every $t \neq 0$. Thus Proposition \ref{RegProp} shows that $\om$ cannot be regular. 
\end{proof}

As explained in \cite{GH:88} and \cite{Grundling:06}, the incompatibility of regularity with the Dirac property is resolved by observing that the discontinuity of the GNS representation in the above proof is caused by an element $W(z_0) \in \W \setminus \mathcal{O}$, with $\mathcal{O}$ defined in \eqref{DefO}, while it is the quotient $\mathcal{O} / \mathcal{D}$ that has to be considered as \emph{$C^*$-algebra of physical observables} and the Dirac states on $\mathcal{O}$ are in $1$-$1$-correspondence with the set of all states on $\mathcal{O} / \mathcal{D}$. Thus there exist regular GNS representations of the physical observables defining there also the self-adjoint field operators.

\section{Normal states on von Neumann algebras generated from representations of $\W$}

Let $\pi \col \W \to B(H)$ be a representation of $\W$. Since $\W$ is simple (\cite[Theorem 3.7(i)]{Slawny:72}), we have $\ker(\pi) = \{ 0\}$  and hence $\pi$ is an isometric $^*$-isomorphism of $\W$ with $\pi(\W)$ (the latter as a $C^*$-subalgebra of $B(H)$). 
We denote by $\W_\pi$ the von Neumann algebra generated from $\pi(\W) \subseteq B(H)$ which is obtained as the double commutant $\pi(\W)''$. Moreover, $\pi(\W)''$ agrees with the closure of $\pi(\W)$ in the strong or the weak operator topology or also in the weak* topology of $B(H)$ (cf.\ \cite[Proposition 8.3, Theorem 12.3, and Proposition 21.8]{Conway:00}). Let us call $\W_\pi$ the \emph{Weyl--von Neumann algebra associated with the representation $\pi$}.
In case $\pi$ is irreducible, we have $\W_\pi = B(H)$ (\cite[Theorem 32.6]{Conway:00}). If $\pi_u$ is the universal representation of $\W$, which is defined as the sum over all GNS representations of $\W$, then $\W_{\pi_u}$ is the so-called \emph{enveloping von Neumann algebra} of $\W$ (cf.\ \cite[Section 3.7]{Pedersen:18}).

\begin{remark} (i) The uniqueness result for the Schr\"odinger representation in case of finite dimensional $S$ can be put into the context of separable factor representations (\cite[Section 2]{Slawny:72}). These are representations $\pi$ of $\W$ on separable Hilbert spaces such that the von Neumann algebra $\W_\pi$ is a \emph{factor}, i.e., has the trivial one-dimensional center consisting of scalar multiples of the identity. Von Neumann algebras that are factors, can be distinguished by the so-called types I, II, and III (cf.\ \cite[Section 6.5]{KR:V2}, in particular, \cite[Corollary 6.5.3]{KR:V2}). 
It turns out that in the situation described previously, the factor $\W_\pi$ is always of type I; it is thus isomorphic to the full algebra of bounded operators on some Hilbert space (\cite[Theorem 6.6.1]{KR:V2}). On the other hand, Slawny used in \cite[Section 3]{Slawny:72} explicitly a type II representation of the Weyl relations to define the Weyl algebra as a $C^*$-algebra.

\noindent (ii) In case of the $\theta$-deformed Weyl algebra $\mathcal{A}_\theta$ on the two torus, mentioned in Remark \ref{Wtheta}, the types of operator algebras generated in physically relevant GNS representations in \cite{BNS:91} turned out to be dependent on $\theta$, namely, they produced von Neumann algebras of type I for rational $\theta$, but a type II factor for irrational values of $\theta$. Using a characterization of so-called \emph{tame} discrete groups, one can show that the von Neumann algebra generated in any representation of  $\mathcal{A}_\theta$ 
is always of type I, if and only if $\theta$ is rational (cf.\ \cite{Hoermann:91}).
\end{remark}

Recall that the weak* topology on $B(H)$ is defined by the seminorms $C \mapsto |\text{tr}(C T)|$, where $T$ varies in the set of trace class operators on $H$. \emph{Normal states} on a von Neumann algebra respect limits of increasing nets of hermitian operators in the strong operator topology and may be characterized as the weak* continuous states, or alternatively, as those given in the form $C \mapsto \text{tr}(C T)$ with a positive trace class operator $T$ such that $\text{tr}(T) = 1$ (\cite[Theorem 46.4]{Conway:00}). 

\begin{lemma}
If the representation $\pi$ is regular, then any normal state $\mu$ on $\W_\pi$ defines a regular state $\om := \mu \circ \pi$  on $\W$.
\end{lemma}
\begin{proof}
For arbitrary $z \in S$,  $g_z \col t \mapsto \pi(W(t z))$ maps into the bounded subset of unitary operators $ U(H) \subseteq B(H)$, on which the weak operator topology agrees with the weak* topology (\cite[Proposition 20.1(b)]{Conway:00}). Thus, $g_z$ is weak* continuous and the composition with the normal state $\mu$ then gives the continuous map $t \mapsto \om(W(tz)) = \mu(\pi(W(tz)))$, which proves regularity of $\om$ by Proposition \ref{RegProp}. 
\end{proof}
In the situation of the above lemma, the normal state $\mu$ can also be considered to induce a regular state on $\widetilde{\W} := \pi(\W)$, if we consider $\widetilde{W}(z) := \pi(W(z))$ ($z \in S$) as the Weyl generators, because $t \mapsto \mu(\widetilde{W}(tz)) =  \mu(\pi(W(tz)))$ is continuous. The following result is immediate from the lemma and Theorem \ref{ThmNonreg}.
\begin{corollary} Let $\pi \col \W \to B(H)$ be a regular representation and $\W_\pi$ be the von Neumann algebra generated from $\pi(\W)$. 
If $L \neq \{ 0 \}$, then there is no normal Dirac state $\mu$ on $\W_\pi$, i.e., there exists no normal state $\mu$ on $\W_\pi$ such that $\mu(\pi(W(y))) = 1$ holds for every $y \in L$.
\end{corollary}

The following example is inspired by a Dirac state that has been applied in  \cite{NT:92} to implement a gauge condition in Quantum Electrodynamics.

\begin{example}\label{PGauge2} We take up Example \ref{PGauge}, where $S = Q \times Q$ with a real Hilbert space $Q$ and $L := \{ 0 \} \times Q$. We apply \cite[Proposition 3.1]{Petz:90} to show that there exists a state $\om_0$ on $\W$ such that
$$
      \forall (z_1,z_2) \in Q \times Q \col \quad \om_0(W(z_1, z_2)) = \left. \begin{cases} 0 & \text{if } z_1 \neq 0,\\ 1 & \text{if } z_1 = 0.\end{cases} \right\} =: g_0(z_1,z_2).
$$
We have to verify that the function $g_0 \col S \to \C$ satisfies $g_0(0) = 1$ (which is obvious) and that the map $h \col S \times S \to \C$ with $h (x,y) := g_0(x - y) \exp(- i \be(x,y))$ defines a \emph{positive (semi)definite kernel}, i.e., for all $n \in \N$, $x^{(1)}, \ldots x^{(n)} \in S$, and $c_1, \ldots, c_n \in \C$, the inequality $\sum_{j,k=1}^n c_j \ovl{c_k} h(x^{(j)},x^{(k)}) \geq 0$ holds: 
Writing $x^{(j)} = (x^{(j)}_1, x^{(j)}_2)$ ($j=1, \ldots, n$) and $I := \{ (j,k) \in \{1,\ldots,n\}^2 \mid  x^{(j)}_1 = x^{(k)}_1\}$ we have $g_0(x^{(j)} - x^{(k)}) = 1$, if $(j,k) \in I$, and $g_0(x^{(j)} - x^{(k)}) = 0$, if $(j,k) \notin I$.  Recalling \eqref{MainBeta}, we obtain (interchanging $x^{(j)}_1$ with $x^{(k)}_1$ for $(j,k) \in I$ at the second equality)
\begin{multline*}
   \sum_{j,k=1}^n c_j \ovl{c_k}\, h(x^{(j)},x^{(k)}) = \sum_{j,k=1}^n c_j \ovl{c_k}\,  g_0(x^{(j)} - x^{(k)}) e^{- i \be(x^{(j)},x^{(k)})}\\
   = \sum_{(j,k) \in I} c_j \ovl{c_k}\,  e^{- i \be((x^{(k)}_1, x^{(j)}_2),(x^{(j)}_1, x^{(k)}_2))}
   = \sum_{(j,k) \in I} c_j \ovl{c_k}\,  e^{\frac{i}{2} (x^{(j)}_1 | x^{(j)}_2)} e^{ - \frac{i}{2} (x^{(k)}_1 | x^{(k)}_2)} 
  =  \sum_{(j,k) \in I}^n d_j \ovl{d_k},  
\end{multline*}
where we put $d_j := c_j \exp(i (x^{(j)}_1 | x^{(j)}_2)/2)$ ($j=1,\ldots,n$). Observe that $I$ defines an equivalence relation on $\{1,\ldots,n\}$, hence we have a partition $\{1,\ldots,n\} = I_1 \cup \cdots \cup I_m$ and obtain
$$
   \sum_{(j,k) \in I}^n d_j \ovl{d_k} = \sum_{l=1}^m \sum_{(j,k) \in I_l \times I_l} d_j \ovl{d_k} =
   \sum_{l=1}^m \sum_{j \in I_l} d_j   \ovl{\sum_{k \in I_l} d_k} = \sum_{l=1}^m \Big| \sum_{j \in I_l} d_j \Big|^2 \geq 0.
$$

Obviously, $\om_0$ is a Dirac state on $\W$ adapted to $L$, since $\om_0(W(0,z_2)) = 1$ for every $z_2 \in Q$. Denote the corresponding GNS representation by $\pi_0 \col \W \to B(H_0)$ and its standard cyclic vector by $\Om_0$. It is obviously not weakly operator continuous, since for any nonzero  $z_1 \in Q$ and $t \neq 0$, $\inp{\Om_0}{\pi_0(W(tz_1,0)) \Om_0} = \om_0(W(t z_1,0)) = 0$, while $\inp{\Om_0}{\pi_0(W(0,0))\Om_0} =   \inp{\Om_0}{\Om_0} = 1$. The Hilbert space $H_0$ is \emph{not separable}, since for any $y = (y_1,y_2), z=(z_1,z_2) \in Q$ with $y_1 \neq z_1$, we have 
\begin{multline*}
    \inp{\pi_0(W(y)) \Om_0}{\pi_0(W(z)) \Om_0} = e^{-i \be(y,z)} \inp{\Om_0}{\pi_0(W(z_1 - y_1, z_2 - y_2))\Om_0}\\ 
    = e^{-i \be(y,z)} \om_0(W(z_1 - y_1, z_2 - y_2)) = 0.
\end{multline*}

The non-regular GNS representation $\pi_0$ associated with the Dirac state $\om_0$ gives $\pi_0(\W) \subseteq B(H_0)$ as a $^*$-isomorphic isometric image of the Weyl algebra in the form of operators on the non-separable Hilbert space $H_0$. Let $\W_0$ denote the Weyl--von Neumann algebra generated from $\pi_0(\W)$. 
The cyclic vector $\Om_0 \in H_0$ induces the corresponding vector state $\mu$ on $B(H_0)$, given by $\mu(A) := \inp{\Om_0}{A \Om_0}$ for every $A \in B(H_0)$, which is clearly a normal state. Let $\mu_0$ denote the restriction of $\mu$ to $\W_0$. 

By construction, we have $\mu_0 \circ \pi_0 = \om_0$ and therefore 
\begin{center}
$\mu_0$ is a normal Dirac state on $\W_0$ adapted to $L$. \end{center}
Note that the normal state $\mu_0$ on $\W_0 \supseteq \pi_0(\W)$ certainly induces a non-regular state on the $C^*$-algebra $\pi_0(\W)$, since we know that $t \mapsto \mu_0(\pi_0(W(tz))) = \om_0(W(tz))$ is discontinuous, if $z \in S \setminus L$.
\end{example}

We briefly investigate general properties of normal states over Weyl--von Neumann algebras generated from, not necessarily regular, representations. In a representation $\pi \col \W \to B(H)$ and for any $\xi \in H$, we have the vector functional $\nu_\xi \col B(H) \to \C$,  $A \mapsto \inp{\xi}{A \xi}$. From these functionals we may derive the functions $f_\xi \col S \to \C$, defined by 
$$
  \forall \xi \in H, \forall z \in S\col \quad  f_\xi(z) := \nu_\xi(\pi(W(z))) = \inp{\xi}{\pi(W(z))\xi}. 
$$   
Since $|f_\xi(z)| = |\inp{\xi}{\pi(W(z))\xi}| \leq \Norm{\xi}^2 \Norm{\pi(W(z))} =  \Norm{\xi}^2$, every $f_\xi$ belongs to the Banach space $(F_b(S), \norm{\ }{\infty})$  of  bounded functions $S \to \C$.  

\begin{theorem}\label{LimVectSt} Let $\pi \col \W \to B(H)$ be a representation and $\W_\pi$ be the von Neumann algebra generated from $\pi(\W)$.
 If $\mu$ is a normal state over $\W_\pi$ and $h \col S \to \C$ is defined by $h(z) := \mu(\pi(W(z)))$ ($z \in S$), then $h$ belongs to the closure $V_\pi$ of $\text{\rm span} \{ f_\xi \mid \xi \in H\}$ in $F_b(S)$.
\end{theorem}
\begin{proof} By \cite[Theorem 7.1.12]{KR:V2}, there exists a countable orthonormal family $(\xi_n)_{n \in \N}$ of vectors in $H$, such that $\mu(A) = \sum_{n \in \N} \nu_{\xi_n}(A)$ for every $A \in \W_\pi$, where the convergence of the series is uniformly in $A$ varying in any  bounded set of operators. Since $\Norm{\pi(W(x))} = 1$, we deduce that $h$ is the uniform limit of the partial sums $\sum_{n=1}^m f_{\xi_n} \in F_b(S)$ as $m \to \infty$.
\end{proof}

Note that we have a natural topology on $S$, if $S = Q \times Q$ with a real Hilbert space $Q$.  In  case of a regular representation $\pi \col \W \to B(H)$, every function $f_\xi \col S \to \C$ is continuous. Therefore, $\text{\rm span} \{ f_\xi \mid \xi \in H\}$ is contained in the space $C_b(S)$ of bounded continuous functions $S \to \C$ and we obtain the following direct consequence.
\begin{corollary}\label{CorCb} If $S = Q \times Q$ with a real Hilbert space $Q$ and the representation $\pi$ in Theorem \ref{LimVectSt} is regular, then $h \in V_\pi \subseteq C_b(S)$.
\end{corollary}

\begin{example}\label{ExC0} (i) Let $S = \R^{2n}$ and $\pi$ be the irreducible \emph{Schr\"odinger representation} on $H=L^2(\R^n)$ with the cyclic vector $\Om$ (a Gaussian on $\R^n$; see \cite[Example 5.2.16]{BR:V2}) satisfying $f_\Om(z) = \inp{\Om}{\pi(W(z))\Om} = \exp(-\Norm{z}^2/4)$.  We observe that $f_\Om$ belongs to the space $C_0(\R^{2n})$ of continuous functions vanishing at infinity. We claim that $V_\pi \subseteq C_0(\R^{2n})$, which implies then that also the function $z \mapsto \mu(\pi(W(z)))$belongs to $C_0(\R^{2n})$  for any normal state $\mu$ on $\W_\pi$.

To prove $V_\pi \subseteq C_0(\R^{2n})$, we start by recalling that  $\text{span}\{ \pi(W(y)) \Om \mid y \in \R^{2n}\}$ is dense in $H$. Let $\xi = W(x)\Om$ and $\eta = \pi(W(y)) \Om$ with $x, y \in \R^{2n}$. An application of $W(-x) W(z) W(y) = e^{i (\be(z,y) - \be(x,z+y))} W(z-x + y)$ gives  
\begin{multline}\label{fxieta}
f_{\xi,\eta}(z) :=  \inp{\xi}{\pi(W(z))\eta}   =  \inp{\pi(W(x)) \Om}{\pi(W(z) W(y)) \Om}\\ 
  = e^{i (\be(z,y) - \be(x,z+y))} \inp{\Om}{\pi(W(z-x+y)) \Om} = e^{i (\be(z,y) - \be(x,z+y))} f_\Om(z-x+y),
\end{multline} 
which shows that also $f_{\xi,\eta}(z) \to 0$ as $\Norm{z} \to \infty$. Suppose $\zeta \in H$ is approximated by a linear combination of the form $\eta := \la_1 \xi_1 + \ldots \la_m \xi_m$ with $\xi_j = \pi(W(x_j)) \Om$, $x_j \in \R^{2n}$, and $\la_j \in \C$. We have
$f_\eta = \sum_{j,k=1}^m \ovl{\la_j} \la_k f_{\xi_j,\xi_k} \in C_0(\R^{2n})$ and the standard estimate
\begin{multline}\label{estdiff}
    |f_\zeta(z) - f_\eta(z)| = |\inp{\zeta}{\pi(W(z))\zeta} - \inp{\eta}{\pi(W(z))\eta}| \\ \leq 
    | \inp{\zeta}{\pi(W(z))(\zeta -\eta)}| + | \inp{\zeta - \eta}{\pi(W(z))\eta}|
    \leq \Norm{\zeta} \Norm{\zeta - \eta} + \Norm{\zeta - \eta} \Norm{\eta}
\end{multline}
implies that $f_\zeta \in C_0(\R^{2n})$. Therefore, $\text{span} \{ f_\xi \mid \xi \in H\} \subseteq C_0(\R^{2n})$, which completes the proof, since $V_\pi$ is the closure. 

Note that the property that $h(z):= \mu(\pi(W(z)))$ defines a function $h \in C_0(\R^{2n})$ for any normal state $\mu$, provides a direct proof that for a nontrivial subspace $L \subseteq \R^{2n}$ no normal Dirac state can exist on $\W_\pi = B(L^2(\R^n))$, since this would require $h(y) = 1$ for all $y \in L$, which gives a contradiction, if we choose $\Norm{y} \to \infty$.

(ii) Generalizing the above example to infinitely many degrees of freedom in the Weyl algebra, let $S$ be a complex Hilbert space with inner product $(.,.)_\C$ and the symplectic form on the underlying real vector space $S$ be defined by $\be(x,y) := \Im (x,y)_\C$/2. We consider the \emph{Fock representation} $\pi \col \W \to B(H)$ (cf.\ \cite[Subsection 5.2.3]{BR:V2}) which is the irreducible GNS representation, with cyclic vector $\Om$, corresponding to the regular and pure state $\om_F$ on $\W$ with the property
$$
   \forall z \in S \col \quad  \om_F(W(z)) = \inp{\Om}{\pi(W(z))\Om} = e^{- \Norm{z}^2/4}.
$$
According to \cite[Theorem 5.2.14]{BR:V2}, the normal states on $\W_\pi = B(H)$ are exactly those possessing densely defined self-adjoint  number operators in their corresponding GNS representations. We can add some qualitative information about the function $h \col S \to \C$ associated with a normal state $\mu$ via $h(z) := \mu(\pi(W(z)))$, following the lines of argument in (i). In fact, the details of reasoning from (i) in showing that $h(z) \to 0$ as $\Norm{z} \to \infty$ can be taken over without change. However, since $S$ is not locally compact, we can no longer claim that $h$ (or any of the $f_\xi$) belongs to $C_0(S)$,  if the latter is defined as the $\norm{\ }{\infty}$-closure of the subspace of continuous functions with compact support, which then yields $C_0(S) = \{ 0\}$. By Corollary \ref{CorCb}, we know a priori that $V_\pi \subseteq C_b(S)$, because the Fock representation is regular, and we can still say that $h \in V_\pi \subseteq C_{b0}(S) := \{ f \in C_b(S) \mid \forall \eps > 0\, \exists R \geq 0\, \forall z \in S, \Norm{z} \geq R \col |f(z)| \leq \eps \}$. 
\end{example}

\section{States on the Weyl algebra $\W$ with parameter space $L^2(\R^n)$}

\subsection{Functions induced from states on the Weyl algebra}

We focus now on the situation of our main example class, where $S$ is the complexification of a real Hilbert space $Q$ and the symplectic form is defined as in \eqref{MainBeta}. We saw in Theorem \ref{ThmNonreg} that (nontrivial) Dirac states $\om$ on the Weyl algebra $\W$ then correspond to discontinuous functions $S \to \C$, $z \mapsto \om(W(z))$. By Theorem \ref{LimVectSt}, normal states over the Weyl--von Neumann algebra $\W_\pi$ in any representation $\pi$ belong to the closed subspace $V_\pi$ generated from vector functionals whithin the Banach space of bounded functions $F_b(S, \C)$. In general, $V_\pi$ is not contained in the subspace of continuous bounded functions $C_b(S)$, as could be seen at the end of Example \ref{PGauge2}, but for GNS representations associated with states inducing measurable functions on $S$ we can guarantee that all of $V_\pi$ consists of Borel measurable functions on $S$. 
\begin{proposition} Let $\pi$ be the GNS representation of $\W$ corresponding to the state $\om$ and let $g \col S \to \C$ denote the function $z \mapsto \om(W(z))$. If $g$ is Borel measurable, then $V_\pi$ (as defined in Theorem \ref{LimVectSt}) is contained in the space of bounded Borel measurable complex functions on $S$.
\end{proposition}
\begin{proof} Since $V_\pi$ consists of uniform limits of linear combinations of the functions $f_\xi$ ($\xi \in H$), it suffices to show that every $f_\xi$ is Borel measurable. Reasoning as in Example \ref{ExC0}, more precisely, as in the paragraph containing Equations \eqref{fxieta} and \eqref{estdiff}, everything boils down to the measurability of the following expression with respect to $z \in S$, where $\Om$ denotes the cyclic vector of the GNS representation and $\xi = \pi(W(x)) \Om$, $\eta = \pi(W(y))\Om$  with $x,y \in S$:
$$
   \inp{\xi}{\pi(W(z)) \eta}\\ 
  = e^{i (\be(z,y) - \be(x,z+y))} \inp{\Om}{\pi(W(z-x+y)) \Om} = e^{i (\be(z,y) - \be(x,z+y))} g(z-x+y).
$$
The exponential factor is continuous, hence Borel measurable. Since translation is a homeomorphism $S \to S$, the Borel $\sig$-algebra in $S$ is translation invariant. Therefore, the measurability of $g$ implies that of the translate $g(. -x + y)$ and the proof is complete.
\end{proof}

We come now back to the Dirac state $\om_0$ with GNS representation $\pi_0 \col \W \to B(H_0)$ constructed in Example \ref{PGauge2}. By the above proposition,  for every $z \in S$ and $\xi, \eta \in H_0$, the map $t \mapsto \inp{\xi}{\pi_0(W(tz))\eta}$ is measurable $\R \to \C$, though typically discontinuous. We remark in passing that this provides an alternative reason for the non-separability of $H_0$, since separability of $H_0$ would imply strong continuity of that same map (\cite[Chapter X, Theorem 5.4]{Conway:90}). We will now show that $\om_0$ is approximated by regular states in the sense of pointwise convergence. 
\begin{proposition} The Dirac state $\om_0$ from Example \ref{PGauge2} is the weak* limit of the sequence $(\om_l)_{l \in \N}$ of regular states  given by
\beq\label{quasifrei}
   \om_l(W(z_1,z_2)) = e^{-\frac{l^2}{4}\Norm{z_1}^2 - \frac{1}{4 l^2} \Norm{z_2}^2} \quad (z_1, z_2 \in Q).
\eeq
\end{proposition}
\begin{proof} Every $\om_l$ is an example of a quasifree, hence regular, state by \cite[Theorem 3.4 and Proposition 3.5]{Petz:90}. 
The sequence $(\om_l)_{l \in \N}$ of states on the $C^*$-algebra $\W$ clearly has the uniform operator norm bound $1$ and the subspace $D := \text{span}\{ W(z) \mid z \in S \}$ is norm dense in $\W$. It therefore suffices to show the convergence to $\om_0$ on every element $W(z)$ ($z \in S$), because then convergence on $D$ is clear and we obtain that $(\om_l(A))_{l \in \N}$ is a Cauchy sequence for any $A \in \W$ upon approximating $A$ by some $B \in D$ and using the standard estimate $|\om_l(A) - \om_m(A)| \leq |\om_l(A) - \om_l(B)| + |\om_l(B) - \om_m(B)| + |\om_m(B) - \om_m(A)| \leq 2 \Norm{\om_l} \Norm{A - B} + |\om_l(B) - \om_m(B)|$.
 
If $z_1 \neq 0$, then $-\frac{l^2}{4}\Norm{z_1}^2 - \frac{1}{4 l^2} \Norm{z_2}^2 \to - \infty$ as $l \to \infty$, hence $\om_l(W(z)) \to 0 = \om_0(W(z))$. If $z_1 = 0$, then obviously $\om_l(W(z)) \to 1 =  \om_0(W(z))$ as $l \to \infty$.
\end{proof}

Note tat in the previous statement, $\om_1$ is the Fock state denoted by $\om_F$ in Example \ref{ExC0}(ii) and the state $\om_l$ is a   rescaling of it stemming from the symplectomorphism $(z_1,z_2) \mapsto (l z_1, z_2/l)$ on $S$.

For the remainder of this section we will consider more specifically the case with $S = L^2(\R^n)$ and symplectic form $ \Im \inp{.}{.}/2$ as in Example \ref{Specases}(ii). The function $g \col L^2(\R^n) \to \C$ \emph{induced by a state} $\om$ on $\W$ is given by
\beq\label{indfun}
    g(\psi) := \om(W(\psi)) \quad  \forall \psi \in L^2(\R^n).
\eeq
Clearly, $g$ is a nonlinear bounded function since $|g(\psi)| = |\om(W(\psi))| \leq \Norm{W(\psi)} = 1$ and, in general, not continuous as can be seen with $g_0(\psi) := \om_0(W(\psi))$ induced by the state $\om_0$ as given in Example \ref{PGauge2}: Here, $L = \{ 0 \} \times Q$ can be identified with the subspace of $L^2$-functions having values in $i \R$; thus we have 
\beq\label{valg0}
\text{$g_0(\psi) = 0$, if $\psi \in L^2(\R^n)$ has non-vanishing real part, and $g_0(\psi) = 1$ otherwise;}
\eeq
choose $\psi_0 \in L^2(\R^n)$ real-valued and non-zero, then we have $\frac{1}{k} \psi_0 \to 0$ as $k \to \infty$ and $g_0(\frac{1}{k}\psi_0) = 0$ for every $k$, but $g_0(0) = 1$.

Let $g_l$ ($l \in \N$) denote the function on $L^2(\R^n)$ associated with the quasifree state \eqref{quasifrei} $\om_l$, i.e.,
\beq\label{quasig}
   g_l(\psi) = e^{- \frac{l^2}{4}\Norm{\Re \psi}^2 - \frac{1}{4 l^2} \Norm{\Im \psi}^2} \quad (\psi \in L^2(\R^n)).
\eeq
Every $g_l$ is continuous $L^2(\R^n) \to \C$, even infinitely often Fr\'{e}chet-differentiable, since the exponent is just the restriction of a continuous $\R$-bilinear form to the diagonal in $L^2(\R^n) \times L^2(\R^n)$ (\cite[8.12.9]{Dieudonne:V1}). The sequence $(g_l)_{l \in \N}$ converges to $g_0$ pointwise on $L^2(\R^n)$, which is essentially a repetition of the weak* convergence of $(\om_l)_{l \in \N}$ shown above. 

\begin{remark}\label{remcol} We may restrict $g_l$ and $g_0$ to the subspace of test functions $\D(\R^n)$ and then ask whether these can be considered to be generalized functions on $\R^n$ in the sense of Colombeau's approach in \cite{Colombeau:84}, which is based on smooth functions $\D(\R^n) \to \C$. As explained in \cite{GKOS:01}, the original notion of differentiability can (on $\D(\R^n)$ and $\E'(\R^n)$) be equivalently replaced by smoothness in the sense of the so-called convenient setting, e.g., described in \cite{KM:97}. The latter can be checked by asking whether smooth curves are mapped into smooth curves. We immediately see that $g_0$ is not smooth in that sense  (picking a real-valued $0 \neq \vphi_0 \in \D(\R)$, the curve $t \mapsto t \vphi_0$ is smooth $\R \to \D(\R^n)$, while $t \mapsto g_0(t \vphi_0)$ is not smooth at $t = 0$). On the other hand, every $g_l$ is smooth, since any smooth curve into $\D(\R^n)$ yields also a smooth curve into $L^2(\R^n)$ in the norm sense and thus the Fr\'{e}chet-differentiablity of $g_l$ implies smoothness of the image curve into $\C$. For $g_l$ to define a Colombeau generalized function as an equivalence class $[g_l] \in \G(\R^n)$, one has to test moderateness as $\eps \to 0$ upon inserting derivatives of $\eps$-scaled delta-nets from $\D(\R^n)$ into derivatives of $g_l$ and then map to the quotient modulo the functions with rapidly vanishing $\eps$-tests. In all the testing calculations for $g_l$ we obtain a common overall factor of the form $\exp(- c \eps^{-n})$, since $\norm{\phi_{\eps,x}}{L^2}^2 = \eps^{-n} \norm{\phi}{L^2}$, if $\phi_{\eps,x}(y)  := \phi((y-x)/\eps) \eps^{-n}$. Therefore, we obtain uniform upper bounds $O(\eps^m)$ as $\eps \to 0$ for every $m \in \N$ and this shows that $[g_l] = [0]$ in $\G(\R^n)$. In fact, this  reasoning applies to any function $g$ of the form $\vphi \mapsto \exp(-\inp{\vphi}{C \vphi}/2)$, where $C$ is a positive (covariance) operator (\cite[Section A.4]{GJ:87}). These functions correspond to (inverse) Fourier transforms of  Gaussian measures on distribution spaces (cf.\  Remark \ref{exGauss}). 
\end{remark}

We recall from \cite[Proposition 3.1]{Petz:90}, see also Example \ref{PGauge2}, that for a given function $g \col S \to \C$ with $g(0) = 1$ there exists a unique state $\om$ on $\W$ with the property
$$
      \forall z \in S\col \quad \om(W(z))  = g(z),
$$
if and only if  the map $h \col S \times S \to \C$, defined by
\beq\label{kernfun}
h (x,y) := g(x - y) \exp(- i \be(x,y)), 
\eeq
is a positive kernel in the sense that for all $n \in \N$, $x^{(1)}, \ldots x^{(n)} \in S$, and $c_1, \ldots, c_n \in \C$, 
\beq\label{poskern}
    \sum_{j,k=1}^n c_j \ovl{c_k} h(x^{(j)},x^{(k)}) \geq 0.
\eeq
In the sequel, we will mostly consider the restrictions of scalar functions on $S = L^2(\R^n)$ to the subspace $\S(\R^n)$ of Schwartz functions.

\subsection{The case of continuous induced functions}
For the specific class of continuous functions $g \col L^2(\R^n) \to \C$, the above correspondence with states on $\W$ has an additional relation with measures on the space of temperate distributions $\S'(\R^n)$, which is based on the so-called 
Bochner-Minlos theorem (cf.\ \cite[Theorem A.6.1]{GJ:87}, \cite[Chapter 4, Section 4]{GV4}, or \cite[Section 7.13]{Bogachev:07}), which we briefly recall: To simplify notation, let us write $\S$, $\S'$ in place of $\S(\R^n)$, $\S'(\R^n)$, respectively, and denote by $\S_\R$ the subspace of \emph{real-valued} functions in $\S$ with the (real) dual space $\S_\R'$. 
The \emph{generating functional} of a regular Borel probability measure $\nu$ on $\S'_\R$ is the function $F \col \S_\R \to \C$, defined by
\beq\label{genfun}
     F(\vphi) := \int_{\S'_\R}  e^{i \dis{u}{\vphi}}  \, d\nu(u) \quad (\vphi \in \S_\R).
\eeq
The basic properties of $F$ are (i) F(0) = 1, (ii) continuity, and (iii) positive (semi-)definiteness, i.e., $0 \leq \sum_{j,k=1}^N \ovl{c_j} c_k F(\vphi_j - \vphi_k)$ for all $N \in \N$, $c_j \in \C$, $\vphi_j \in \S_\R$. Conversely, given a Function $F \col \S_\R \to \C$ with properties (i)-(iii), there exists  a unique regular Borel probability measure $\nu$ on $\S'_\R$ such that \eqref{genfun} holds, i.e., $F$ is a kind of inverse Fourier transform of $\nu$.  

Now let $g \col L^2(\R^n) \to \C$ be continuous with $g(0) = 1$ and such that it defines a positive kernel in the sense of \eqref{kernfun} and \eqref{poskern}.
We  consider $F_1 \col \S_\R \to \C$ obtained by restriction of $g$ to the subspace of real-valued Schwartz functions. Since $\beta(\vphi,\psi) = \Im \inp{\vphi}{\psi}/2 = 0$ for all $\vphi, \psi \in \S_\R$,  \eqref{poskern} implies that $F_1$ is positive (semi-)definite in the sense required for the Bochner-Minlos theorem. Thus, there is a unique regular Borel measure $\nu_1$ on $\S_\R'$ with $\nu_1(\S'_\R) = 1$ such that 
\beq\label{realint}
   \forall \vphi \in \S_\R \col \quad g(\vphi) = F_1(\vphi) = \int_{\S_\R'}  e^{i \dis{u}{\vphi}}  \, d\nu_1(u)
\eeq
and we may formulate the following result.
\begin{proposition} If $\om$ is a state on the Weyl algebra $\W$ over $L^2(\R^n)$ with continuous induced map $g \col L^2(\R^n) \to \C$, $\vphi \mapsto \om(W(\vphi))$, then there is a unique regular Borel probability measure $\nu_1$ on $\S_\R'$ such that \eqref{realint} holds.
\end{proposition}

\begin{remark}\label{exGauss} (i) For each quasifree state \eqref{quasifrei} with induced function given in \eqref{quasig}, there corresponds a Gaussian measure on $\S'_\R$ (\cite[Theorem A.4.6]{GJ:87} or \cite[Chapter 4, Example in Subsection 4.1]{GV4}) in the sense of  representation \eqref{realint}; this holds more generally for any Gaussian-type function given in terms of a covariance operator as those described at the end of Remark \ref{remcol}.

\noindent (ii) As explained in \cite[Chapter IV, Subsection 5.4]{GV4}, one may use $L^2(\S'_\R,\nu_1)$ as a representation Hilbert space for the canonical commutation relations, i.e., of the Weyl algebra $\W$. For any $\vphi \in \S_\R$, the action of $W(\vphi)$ is multiplication by the function $u \mapsto \exp(i \dis{u}{\vphi})$, while the action of $W(i \vphi)$ on a function $f \in L^2(\S'_\R,\nu_g)$ produces the function $u \mapsto a_\vphi(u) f(u + R_\vphi)$, where $R_\vphi \in \S'_\R$ is defined by $\psi \mapsto \beta(\psi,i\vphi)$ and $a_\vphi \col \S'_\R \to \C$ results from (a pre-defined) action of $W(i \vphi)$ on $1 \in L^2(\S'_\R,\nu_1)$, subject to the functional equation $a_{\vphi_1 + \vphi_2}(u) = a_{\vphi_1}(u) a_{\vphi_2}(u + R_{\vphi_1})$.
\end{remark}

An analogous construction starting from $F_2 \col \S_\R \to \C$, defined by $F_2(\vphi) := g(i \vphi)$ for all $\vphi \in \S_\R$, gives a regular Borel probability measure $\nu_2$ on $\S'_\R$ such that
\beq\label{comint}
   \forall \vphi \in \S_\R \col \quad g(i \vphi) = F_2(\vphi) = \int_{\S_\R'}  e^{i \dis{v}{\vphi}}  \, d\nu_2(v).
\eeq
In case $g$  is of product form in the sense that $g(\psi) = g(\Re \psi + i \Im \psi) = g_1(\Re \psi) g_2(i \Im \psi)$ holds for all $\psi \in \S$, one can easily combine the two measures $\nu_1$ and  $\nu_2$ obtained for the real and imaginary part separately. Fubini's theorem yields a straightforward interpretation in terms of the product measure $\nu_1 \otimes \nu_2$ by
$$
   g(\psi) = g_1(\Re \psi) g_2(i \Im \psi)  =
   \int_{\S'_\R \times \S'_\R}  e^{i (\dis{u}{\Re \psi} + \dis{v}{\Im \psi})}  \,d (\nu_1 \otimes \nu_2) (u,v).
$$
This applies, for example, directly to Gaussian-type functions $g$. If $g \in C_b(\S_\R) \otimes C_b(i \S_\R)$, say $g = \sum_{k=1}^N g_{1,k} \otimes g_{2,k}$ with $g_{1,k} \in C_b(\S_\R)$ and $g_{2,k} \in C_b(i \S_\R)$ and corresponding measures $\nu_{1,k}$ and $\nu_{2,k}$ on $\S_\R'$ ($k=1,\ldots,N$), then we obtain with $\nu := \sum_{k=1}^N \nu_{1,k} \otimes \nu_{2,k}$ the representation
$$
    g(\psi) =  \int_{\S'_\R \times \S'_\R}  e^{i (\dis{u}{\Re \psi} + \dis{v}{\Im \psi})}  \,d\nu (u,v).
$$
It is tempting to think of a situation where $g \in C_b(\S)$ is approximated in norm by a sequence of functions $g^{(m)} \in C_b(\S_\R) \otimes C_b(i \S_\R)$ and each $g^{(m)}$ is given by a Fourier integral of a measure $\nu^{(m)}$ on $\S' \isom \S'_\R \times \S'_\R$ as above. One might then hope that $\nu^{(m)}$ converges to a measure $\nu$ on $\S'$  in an appropriate sense and that a similar Fourier integral representation for $g$ in terms of $\nu$ could be achieved. In particular with the current context, where the underlying topological groups are not locally compact, the author is not aware of results to settle these questions.

Suppose now that $\nu$ is a Borel probability measure on $\S_\R'$ satisfying a \emph{polynomial growth condition} in the sense that any function of the form 
\beq\label{polgrowth}
    u \mapsto \dis{u}{\psi_1} \cdots \dis{u}{\psi_k} \text{ is $\nu$-integrable, where $k \in \N$ and $\psi_j \in \S_\R$ ($j=1,\ldots,k$)}.
\eeq
Let $F$ be given by the formula in Equation \eqref{genfun}, then we claim that $F$ is smooth on the subspace $\D_\R \subseteq \D(\R^n)$ of real-valued test functions in the sense of the convenient setting as described briefly in Remark \ref{remcol} (this smoothness holds also on $\S_\R$).  In fact, let $c \col I \to \D_\R$ be a smooth curve defined on some interval $I \subseteq \R$ and put $f(t,u) := \exp(i \dis{u}{c(t)})$ for every $t \in I$ and $u \in \S'_\R$. We may apply the standard theorems about integrals with dependence on real parameters, since $u \mapsto f(t,u)$ is $\nu$-integrable for every $t \in I$, the function $t \mapsto f(t,u)$ is smooth $I \to \R$ for every $u \in \S'_\R$, and every derivative $\d_t^k f(t,u)$ ($k \in \N_0$) is bounded by some $\nu$-integrable function.
Moreover, inserting $\eps$-scaled and $x$-shifted delta-nets for $\vphi$ as used for testing moderateness in Colombeau's theory, it is easily seen from the structure of higher derivatives, namely as polynomial times exponential with purely imaginary exponent, that we always obtain uniform upper bounds in terms of polynomials in $1/\eps$ for $x$ varying in compact subsets of $\R^n$. Thus, $F$ as given in \eqref{genfun} does define a (real) Colombeau generalized function on $\R^n$. 
To summarize, we have shown the following statement. 
\begin{proposition} Let $\om$ be a state on the Weyl algebra $\W(L^2(\R^n),\be)$ with continuous induced function $g \col L^2(\R^n) \to \C$ and such that the measure $\nu_1$ on $\S_\R'$ in the representation \eqref{realint} satisfies polynomial growth conditions as in \eqref{polgrowth}. Then the restriction of $g$ to real test functions defines a Colombeau generalized functions on $\R^n$. The analogous statement holds for the restriction of $g$ to imaginary parts, if also the measure $\nu_2$ in \eqref{comint} satisfies polynomial growth conditions.
\end{proposition}

\subsection{Harmonic analysis for more general types of induced functions}

Here we investigate how to go beyond continuity of the function induced by a state on the Weyl algebra over $L^2(\R^n)$ and try to find some structure in case of a \emph{Borel measurable} function $g \col L^2(\R^n) \to \C$ that satisfies $g(0) = 1$ and (\ref{kernfun}-\ref{poskern}). In particular, this would be interesting for the discontinuous function $g_0$ constructed in Example \ref{PGauge2} from a Dirac state and which is the pointwise limit of a sequence of continuous functions corresponding to regular states on the Weyl algebra. 

To compare with the finite-dimensional situation, recall that on $\R^d$ any Lebesgue measurable positive definite function coincides almost everywhere with a \emph{continuous} positive definite function (\cite[Theorem 3.10.20]{Bogachev:07}) and is thus the inverse Fourier transform of some nonnegative measure by the classical Bochner theorem.
According to \cite[Corollary 7.13.8]{Bogachev:07}, a function $g \col L^2(\R^n) \to \C$ with $g(0)=1$ is the generating functional of a regular Borel probability measure on $L^2(\R^n)$, if and only if $g$ is positive definite and continuous with respect to the \emph{Sazonov topology}, which is generated by the seminorms $\vphi \mapsto p_T(\vphi) := \Norm{T \vphi}$, where $T$ is a Hilbert-Schmidt operator. 
We immediately see that the function $g_0$ induced from the Dirac state in Example \ref{PGauge2}  is certainly not continuous on $L^2(\R^n)$ with respect to the Sazonov topology, since the latter is coarser than the norm topology. 

Let us denote the restriction to $\S$ again by $g_0$ and note that we have
\beq\label{g01L}
   g_0 = 1_{L}, 
\eeq  
where $1_L$ denotes the characteristic function of the subspace 
\beq\label{defL}
    L := \{ \psi \in \S \mid \Re \psi = 0\}.
\eeq
Let us once more stress a finite-dimensional analog: 
If  $L_0$  is a subspace of $\R^d$, then the density function $1_{L_0}$ on $\R^d$ corresponds to 
the Euclidean surface measure $\de_{L_0}$ on $L_0$ belongs to $\S'(\R^d)$ and has Fourier transform 
\beq\label{endldimanalog}
\FT{\de_{L_0}} = (2 \pi)^{\dim {L_0}} \de_{{L_0}^\perp}
\eeq 
(\cite[Theorem 7.1.25]{Hoermander:V1}),  a multiple of the surface measure on the Euclidean orthogonal complement. We want to establish a variant of this finite-dimensional result involving the annihilator in the sense of dual spaces in place of the orthogonal complement ${L_0}^\perp$ in $\R^d$.  A question then is how to reasonably define $\FT{g_0} \col \S' \to \C$ and whether this coincides with $h_0 := 1_{L^\perp}$ as a function or measure  concentrated on the annihilator $L^\perp \subseteq \S'$, or rather, whether we can represent $g_0$ as an inverse Fourier transform of $h_0$. We certainly will need to deal also with non-finite measures on $\S$ or $\S'$, because even the simplest classical relation establishing $\de_0$ on $\R^d$ as the inverse Fourier transform of $1$, formally written $\de_0(x) = \int_{\R^d} e^{i x \xi} \, d\xi/(2 \pi)^d$,  involves the non-finite measure $d\xi/(2 \pi)^d$ on $\R^d$.

Since we are looking for an appropriate notion of Fourier transform we might as well consider $(\S,+)$ as a commutative topological group, which is certainly not locally compact. A standard way in harmonic analysis (e.g., \cite[Sections 10.2 and 12.1]{Kirillov:76})) is then to define the \emph{Fourier transform} on the convolution Banach algebra $M(\S)$ of complex (finite) Borel measures on $\S$, equipped with the total variation as norm. If $\mu \in M(\S)$, then its Fourier transform $\FT{\mu}$ is a complex-valued function on the \emph{dual group} $\FT{\S}$ consisting of the (continuous) \emph{characters}, i.e., continuous group homomorphisms from $\S$ to the one-dimensional torus group $\T$. The set $\FT{\S}$ is itself a commutative topological group under pointwise multiplication of characters and equipped with the compact-open topology. The abstract definition (\cite[\pg 12, Equation (1)]{Kirillov:76}) gives
\beq\label{FTgroup}
   \forall \la \in \FT{\S}\col \quad  \FT{\mu}(\la) := \int_{\S} \ovl{\la(\vphi)} \, d\mu(\vphi),
\eeq
but we will be able to rewrite this with more functional analytic notation in a moment. 

\begin{example}\label{dedach} Consider $\de_0 \in M(\S)$, then we obtain 
$$
  \FT{ \de_0}(\la) = \int_{\S} \ovl{\la(\vphi)} \, d \de_0(\vphi) = \overline{\la(0)} = 1 \quad \forall \la \in \FT{\S},
$$
and thus $\FT{\de_0} = 1$.
\end{example}

We also have the so-called \emph{co-Fourier transform} which maps an element $\nu \in M(\FT{\S}\;)$ to the function $\widetilde{\nu} \col \S \to \C$, defined by
\beq\label{coFTgroup}
      \forall \vphi \in \S\col \quad  \widetilde{\nu}(\vphi) := \int_{\FT{\S}} \la(\vphi) \, d\nu(\la). 
\eeq 

Since a character $\la \in \FT{\S}$ has to be continuous $\S \to \T \subseteq \C$ and to satisfy $\la(\vphi + \psi) = \la(\vphi) \la(\psi)$ for all $\vphi, \psi$, we obviously obtain examples of such in the following form: Denote by $\S_0'$ the space of  continuous $\R$-linear functionals $\S \to \R$  and let $u \in \S_0'$; then we define $\la_u \col \S \to \T$ by
$$
    \la_u (\vphi) := e^{i \dis{u}{\vphi}} \quad \forall \vphi \in \S. 
$$
It follows from \cite[Lemma 1]{Smith:52} that the map $u \mapsto \la_u$  defines an algebraic isomorphism of groups $\S_0' \to \FT{\S}$ and, in addition, is a homeomorphism, if $\S_0'$ is equipped with the topology $\tau_c$ of uniform convergence on the compact subsets of $\S$. Since $\S$ is a Montel space, the topology $\tau_c$  coincides with the standard (strong) topology on $\S_0'$ (\cite[Proposition 34.5]{Treves:06}) and we obtain the following statement.
\begin{lemma}\label{smithlem} The dual group $\FT{\S}$ is isomorphic as a topological group to $\S_0'$. 
\end{lemma}

\begin{remark} (i) In the general context of a reflexive locally convex vector space $E$ in place of $\S$, the analog of the above observation is also at the basis of a variant of Pontryagin's classical theorem for locally compact abelian groups. Namely, it is true that the canonical map $\iota$ from $E$ into the dual of the dual group $\FT{E}$, where $\iota(x)(\la) := \la(x)$ for $x \in E$ and $\la \in \FT{E}$, is an isomorphism of topological groups. Somewhat surprisingly, in general, reflexivity of a locally convex vector space is not necessary for the Pontryagin-type reflexivity as a topological group; for example, any Banach space has the Pontryagin property, regardless whether it is reflexive or not (see \cite{Smith:52} and also \cite{CDM-P:12} for a broader overview in the context of topological groups). 

\noindent (ii) We have the $\R$-linear isomorphism $\S' \to \S_0'$ given by $u \mapsto \Re u$.{ \small (Injectivity follows, since $\Re u = 0$ implies $u (\psi) =  i \Im u(\psi)$ for all $\psi \in \S$ and then $i u(\psi) = u(i \psi) = i \Im u(i \psi) = i \Im (i u(\psi)) = i \Re u(\psi) = 0$. Surjectivity is established by showing that for any $w \in S_0'$, the map $u \col \S \to \C$, $u(\psi) := w(\psi) - i w(i \psi)$ is $\C$-linear.)} We could therefore alternatively define a character $\chi_v \in \FT{\S}$ for any $v \in \S'$ by 
$$
   \chi_v(\vphi) := e^{i\Re \dis{v}{\vphi}}  \quad \forall \vphi \in \S
$$
and then obtain instead the isomorphism $\FT{\S} \isom \S'$ of topological groups. However, for the formulae of Fourier transforms below, we prefer to avoid the explicit appearance of the real part in the exponent and will stay with the character space $\S_0' \isom \FT{\S}$ instead.
\end{remark}

Lemma \ref{smithlem} allows us to rephrase the definition of the Fourier transform \eqref{FTgroup} and of the co-Fourier transform \eqref{coFTgroup} with some abuse of notation in the following way: For any element $\mu \in M(\S)$, we have as Fourier transform the function $\F\mu = \FT{\mu} \col \S_0' \to \C$, defined by
\beq\label{FTS}
   \forall u \in \S_0' \col \quad  \F\mu (u) = \FT{\mu}(u) := \int_{\S} e^{- i \dis{u}{\vphi}} \, d\mu(\vphi).
\eeq
If $\nu \in M(\S_0')$, then its co-Fourier transform $\widetilde{\nu} \col \S \to \C$ is given by
\beq\label{coFT}
      \forall \vphi \in \S\col \quad  \widetilde{\nu}(\vphi) := \int_{\S_0'} e^{i \dis{u}{\vphi}} \, d\nu(u). 
\eeq 

\begin{example}\label{oddex} We can partially mimic the typical duality trick from distribution theory by employing the obvious definitions for the action of a finite Borel measure on a bounded Borel measurable function, both considered either on $\S$ or on $\S_0'$. Let $\mu \in M(\S)$ and $\nu \in M(\S_0')$ be arbitrary, then we obtain from Fubini's theorem and upon defining $\FT{\nu}(\vphi) := \widetilde{\nu}(-\vphi)$ the following familiar looking relation
\begin{multline*}
   \dis{\FT{\mu}}{\nu} = \int_{\S_0'} \FT{\mu}(u) \, d \nu(u) =  \int_{\S_0'} \int_{\S} e^{-i \dis{u}{\vphi}}  \, d\mu(\vphi) \, d \nu(u) =
   \int_{\S}  \int_{\S_0'} e^{-i \dis{u}{\vphi}}  \, d \nu(u)  \, d\mu(\vphi)\\
    =  \int_{\S} \FT{\nu}(\vphi)  \, d\mu(\vphi) = \dis{\mu}{\FT{\nu}}.
\end{multline*}
Similarly, defining $\widetilde{\mu}(u) := \FT{\mu}(-u)$, we also have $\dis{\widetilde{\mu}}{\nu} = \dis{\mu}{\widetilde{\nu}}$.
The special case with $\mu = \de_0 \in M(\S)$ gives $\widetilde{\de_0} = 1$ and $\dis{1}{\nu} = \dis{\de_0}{\widetilde{\nu}}$, while we obtain for the case $\nu = \de_0 \in M(\S_0')$ that $\FT{\de_0} = 1$ and $\dis{\mu}{1} = \dis{\FT{\mu}}{\de_0}$. However, these relations do not provide us with a complement to Example \ref{dedach} in terms of a formula like $\de_0 = \widetilde{1}$ or $\de_0 = \FT{1}$, because the  duality pairings used above do not extend to Fourier (co-)transforms of infinite measures or bounded Borel measurable functions. One would need to identify test function spaces on $\S$ and on $\S_0'$ that are mapped into each other by the Fourier (co-)transform and at the same time both allow for a duality pairing with infinite measures or bounded Borel functions. 
\end{example}

\begin{remark}\label{oddbem}
In pursuing the quest for an infinite-dimensional analog of the formula $\widetilde{1} = \de_0$, let us add the following bits of heuristics to the above example: Let $\nu_l$ ($\l \in \N$) be a Gaussian measure on $\S_0'$ with generating functional of the form $\widetilde{\nu_l}(\vphi) = \exp(- l^2 \norm{\vphi}{L^2}^2/4)$ (compare with Remark \ref{exGauss}). On the one hand, for any $\mu \in M(\S)$ and writing $d_0 := 1_{\{ 0\}}$, dominated convergence implies that $\dis{\mu}{\widetilde{\nu_l}} \to \mu(\{0\}) = \dis{\mu}{d_0}$ as $l \to \infty$, thus $\widetilde{\nu_l} \to d_0$. On the other hand, noting that $\nu_l$ has covariance operator $l^2$ times the identity on $\S$ (cf.\ \cite[Appendix to Part I, Section A.4]{GJ:87}), we have
$$
    \forall \vphi, \psi \in \S\col \quad \int_{\S_0'} \dis{u}{\vphi} \dis{u}{\psi} \, d\nu_l(u) = l^2 \inp{\vphi}{\psi}.
$$
This formula suggests that $\nu_l$ may be seen as an approximation to an (unbounded) measure on $\S_0'$ with ``some constant density $c > 0$'', i.e., $\nu_l \approx c$ for large $l$. Combining these two aspects we obtain in a vague sense that $c \, \dis{\mu}{\widetilde{1}} = c\, \dis{\widetilde{\mu}}{1} \approx \dis{\widetilde{\mu}}{\nu_l} = \dis{\mu}{\widetilde{\nu_l}} \to \dis{\mu}{d_0}$, which supports the idea to expect some relation like  $\widetilde{1} \approx d_0/c $.  
\end{remark}

Coming back to the search for an analog of formula \eqref{endldimanalog} in case of the subspace $L \subseteq \S$ defined in \eqref{defL} we consider the topological direct sum decomposition 
$$
   \S = K \oplus L \isom K \times L,
$$ 
where $K := \{ \psi \in \S \mid \Im \psi = 0 \}$. This implies then the decomposition
$$
   \S_0' = L^\perp \oplus K^\perp \isom L^\perp \times K^\perp
$$
in terms of the annihilators, where $V^\perp :=\{ u \in \S_0' \mid  \forall \psi \in V \col \dis{u}{\psi} = 0\}$ for any $V \subseteq \S$. Recall that we had identified $g_0 \col \S \to \C$ in \eqref{g01L} with the characteristic function $1_L$ of $L$ and define also $h_0 := 1_{L^\perp}$ as a function on $\S_0'$. Any of the relations  $h_0 = \FT{g_0}$ or $g_0 = \widetilde{h_0}$ could be considered an analog of \eqref{endldimanalog} and we will argue that they hold at least in some approximate sense.  

Recall that the product $f \cdot \rho$ of a bounded Borel measurable function $f$ with a Borel measure $\rho$ is the measure assigning the value $\int_B f(x) \, d\rho(x)$ to any Borel subset $B$.
\begin{lemma} Let $\mu \in M(\S)$ be a product measure with respect to the decomposition $\S \isom K \times L$, i.e., $\mu = \mu_1 \otimes \mu_2$ with $\mu_1 \in M(K)$, $\mu_2 \in M(L)$. 

\noindent (i) We have $g_0 \mu =  \mu_1(\{0\}) (\de_0 \otimes \mu_2)$.

\noindent (ii)  If $\mu = \de_0 \otimes \mu_2$ with $\mu_2 \in M(L)$, then we obtain the following equation of bounded Borel functions on $\S_0'$:
\beq\label{FTg0}
  \forall (u_1,u_2) \in L^\perp \times K^\perp \col \quad  \FT{(g_0 \mu)}(u_1,u_2) =   \FT{\mu_2}(u_2).
\eeq

\end{lemma}
\begin{proof} (i): Let $B = B_1 \times B_2 \subseteq \S$ with Borel sets $B_1 \subset K$ and $B_2 \subseteq L$. Recalling $1_L(\vphi_1,\vphi_2) = 0$ for $\vphi_1 \neq 0$ and $1_L(0,\vphi_2) =1$, we have
\begin{multline*}
   g_0 \mu (B) = \int_B g_0(\vphi)\, d\mu(\vphi) = \int_{B_2} \int_{B_1} 1_L(\vphi_1,\vphi_2) \,d\mu_1(\vphi_1) d\mu_2(\vphi_2)
   = \int_{B_2}  \mu_1(\{ 0 \}) \de_0(B_1) \,d\mu_2(\vphi_2)\\
   =  \mu_1(\{ 0 \}) \de_0(B_1) \mu_2(B_2) =  \mu_1(\{ 0 \}) (\de_0 \otimes \mu_2)(B),
\end{multline*}
thus, the measures on both sides of the claimed equality agree on a generating family of Borel sets in the product space $K \times L$.

\noindent (ii): We have
\begin{multline*}
   \FT{(g_0 \mu)}(u_1,u_2) = \F{  (\de_0 \otimes \mu_2)}(u_1,u_2) = 
    \int_{K \times L} e^{i (\dis{u_1}{\vphi_1} + \dis{u_2}{\vphi_2})}\, d (\de_0 \otimes \mu_2)(\vphi_1,\vphi_2)\\
   =   \int_K e^{i (\dis{u_1}{\vphi_1}} \, d\de_0(\vphi_1) \int_L e^{i (\dis{u_2}{\vphi_2}} \, d\mu_2(\vphi_2) =
   e^{i \dis{u_1}{0}}\, \FT{\mu_2}(u_2) =   \FT{\mu_2}(u_2).
\end{multline*}
\end{proof}

Let us define the convolution $f * \rho$ or $\rho * f$ of a bounded Borel function $f$ with a complex (finite) Borel measure $\rho$, both defined on $\S_0'$ or both on one of the subspaces $L^\perp$ and $K^\perp$, to be the bounded Borel measurable function given by $u \mapsto \int f(u - v) \, d\rho(v)$. We clearly have for $\de_0 \in M(K^\perp)$ that $\de_0 * \FT{\mu_2} = \FT{\mu_2}$ and that $\rho_1 * 1$ gives the constant function with value $\rho_1(L^\perp)$ for any measure $\rho_1 \in M(L^\perp)$. Denoting the function $(u_1,u_2) \mapsto \FT{\mu_2}(u_2)$ by $1 \otimes \FT{\mu_2}$ we thus have
$$
   (\rho_1 \otimes \de_0) * (1 \otimes \FT{\mu_2}) = (\rho_1 * 1) \otimes (\de_0 * \FT{\mu_2}) = \rho_1(L^\perp) \cdot (1 \otimes \FT{\mu_2})
$$
and deduce therefore from \eqref{FTg0} that
\beq\label{g0h0murho}
    \rho_1(L^\perp) \cdot \FT{(g_0 \mu)} =   (\rho_1 \otimes \de_0) * ( 1 \otimes \FT{\mu_2}) =  (\rho_1 \otimes \de_0) * ( \FT{\de_0} \otimes \FT{\mu_2})
    = (\rho_1 \otimes \de_0) * \FT{\mu}
\eeq
holds for $\mu = \de_0 \otimes \mu_2$ with arbitrary $\mu_2 \in M(L)$ and $\rho_1 \in M(L^\perp)$.

We have $g_0 \mu = \de_0 \otimes \mu_2$ by (i) in the above lemma, hence this can be considered an approximation for a ``surface measure'' $\de_L$  on $L$, if $\mu_2$ is a positive measure with large support (e.g., Gaussian measures as mentioned in Remark \ref{oddbem}). In a similar way, one can derive the equation 
$$
   \rho_1 \otimes \de_0 = h_0 (\rho_1 \otimes \de_0)
$$ 
and this measure resembles an approximation of a ``surface measure'' $\de_{L^\perp}$ on $L^\perp$, if $\supp(\rho_1)$ is large and $\rho_1$ is positive. 
Note that for $\rho := \rho_1 \otimes \de_0$ we obtain $\Norm{\rho} = \rho(\S_0') = \rho_1(L^\perp) \de_0(K^\perp) = \rho_1(L^\perp)$, if $\rho_1$ is positive. From Equation \eqref{g0h0murho} we then immediately obtain the following result.
\begin{proposition} Let $\mu_2 \in M(L)$ and $\rho_1 \in M(L^\perp)$ be finite positive Borel measures and define $\rho := \rho_1 \otimes \de_0$ and $\mu := \de_0 \otimes \mu_2$ as product measures on $L^\perp \times K^\perp \isom \S_0'$ and on $K \times L \isom \S$, respectively. Denote by $g_0$ the characteristic function of $L \subseteq \S$ and by $h_0$ the characteristic function of $L^\perp \in \S_0'$, then we have the relation
\beq\label{closeanalog}
       \Norm{\rho} \cdot \FT{(g_0 \mu)} = (h_0 \rho) * \FT{\mu}.
\eeq
\end{proposition}

Let us finally exploit Equation \eqref{closeanalog} in an attempt to find at least heuristically an analog of Equation \eqref{endldimanalog}. 
Formally, $\mu_2 \approx 1$ yields $g_0 \mu \approx g_0$ and $\FT{\mu} = 1 \otimes \FT{\mu_2} \approx 1 \otimes \de_0$, so that \eqref{closeanalog} then suggests
$$
     \Norm{\rho} \cdot \FT{g_0} \approx (h_0 \rho) * (1 \otimes \de_0) = (\rho_1 \otimes \de_0) * (1 \otimes \de_0) = 
     (\rho_1 * 1) \otimes (\de_0 * \de_0) = \Norm{\rho} 1 \otimes \de_0 \approx  \Norm{\rho} h_0,
$$
which ``implies''
$$
    \FT{g_0} \approx h_0.
$$
Reasoning in similar way, but starting with a relation of the form $\widetilde{(h_0 \rho)} = \rho_2(\{0\}) (\widetilde{\rho_1} \otimes 1)$ for a product measure $\rho = \rho_1 \otimes \rho_2$ and then considering $(g_0 \mu) * \widetilde{\rho}$, one could also obtain a variant of the above proposition and thus a reasonable suggestion of the relation
$$
     g_0 \approx \widetilde{h_0}.
$$

\bibliography{gue}

\newcommand{\SortNoop}[1]{}
\begin{bibdiv}
\begin{biblist}

\bib{Bellissard:86}{incollection}{
      author={Bellissard, J.},
       title={{$K$}-theory of {$C^\ast$}-algebras in solid state physics},
        date={1986},
   booktitle={Statistical mechanics and field theory: mathematical aspects
  ({G}roningen, 1985)},
      series={Lecture Notes in Phys.},
      volume={257},
   publisher={Springer, Berlin},
       pages={99\ndash 156},
         url={https://doi.org/10.1007/3-540-16777-3_74},
      review={\MR{862832}},
}

\bib{BNS:91E}{article}{
      author={Benatti, F.},
      author={Narnhofer, H.},
      author={Sewell, G.~L.},
       title={Errata: ``{A} noncommutative version of the {A}rnol\cprime d cat
  map''},
        date={1991},
        ISSN={0377-9017},
     journal={Lett. Math. Phys.},
      volume={22},
      number={1},
       pages={81},
         url={https://doi.org/10.1007/BF00400380},
      review={\MR{1121851}},
}

\bib{BNS:91}{article}{
      author={Benatti, F.},
      author={Narnhofer, H.},
      author={Sewell, G.~L.},
       title={A noncommutative version of the {A}rnol\cprime d cat map},
        date={1991},
        ISSN={0377-9017},
     journal={Lett. Math. Phys.},
      volume={21},
      number={2},
       pages={157\ndash 172},
         url={https://doi.org/10.1007/BF00401650},
      review={\MR{1093527}},
}

\bib{Bogachev:07}{book}{
      author={Bogachev, V.~I.},
       title={Measure theory. {V}ol. {I}, {II}},
   publisher={Springer-Verlag, Berlin},
        date={2007},
        ISBN={978-3-540-34513-8; 3-540-34513-2},
         url={https://doi.org/10.1007/978-3-540-34514-5},
      review={\MR{2267655}},
}

\bib{BR:V2}{book}{
      author={Bratteli, O.},
      author={Robinson, D.~W.},
       title={Operator algebras and quantum statistical mechanics. 2},
     edition={Second},
      series={Texts and Monographs in Physics},
   publisher={Springer-Verlag, Berlin},
        date={1997},
        ISBN={3-540-61443-5},
         url={https://doi.org/10.1007/978-3-662-03444-6},
        note={Equilibrium states. Models in quantum statistical mechanics},
      review={\MR{1441540}},
}

\bib{CDM-P:12}{article}{
      author={Chasco, M.~J.},
      author={Dikranjan, D.},
      author={Mart\'{\i}n-Peinador, E.},
       title={A survey on reflexivity of abelian topological groups},
        date={2012},
        ISSN={0166-8641},
     journal={Topology Appl.},
      volume={159},
      number={9},
       pages={2290\ndash 2309},
         url={https://doi.org/10.1016/j.topol.2012.04.012},
      review={\MR{2921819}},
}

\bib{Colombeau:84}{book}{
      author={Colombeau, J.~F.},
       title={New generalized functions and multiplication of distributions},
   publisher={North-Holland},
     address={Amsterdam},
        date={1984},
}

\bib{Conway:90}{book}{
      author={Conway, J.~B.},
       title={A course in functional analysis},
     edition={Second Edition},
      series={Graduate Texts in Mathematics},
   publisher={Springer-Verlag, New York},
        date={1990},
      volume={96},
        ISBN={0-387-97245-5},
      review={\MR{1070713}},
}

\bib{Conway:00}{book}{
      author={Conway, J.~B.},
       title={A course in operator theory},
      series={Graduate Studies in Mathematics},
   publisher={American Mathematical Society, Providence, RI},
        date={2000},
      volume={21},
        ISBN={0-8218-2065-6},
      review={\MR{1721402}},
}

\bib{Dieudonne:V1}{book}{
      author={Dieudonn{\'{e}}, J.},
       title={Grundz{\"u}ge der modernen analysis, band 1},
   publisher={Vieweg},
     address={Braunschweig},
        date={1985},
}

\bib{GW:54}{article}{
      author={G{\aa}rding, L.},
      author={Wightman, A.},
       title={Representations of the commutation relations},
        date={1954},
        ISSN={0027-8424},
     journal={Proc. Nat. Acad. Sci. U.S.A.},
      volume={40},
       pages={622\ndash 626},
         url={https://doi.org/10.1073/pnas.40.7.622},
      review={\MR{62957}},
}

\bib{GV4}{book}{
      author={Gel\cprime~fand, I.~M.},
      author={Vilenkin, N.~Ya.},
       title={Generalized functions. {V}ol. 4},
   publisher={AMS Chelsea Publishing, Providence, RI},
        date={2016},
        ISBN={978-1-4704-2662-0},
        note={Applications of harmonic analysis, Translated from the 1961
  Russian original [ MR0146653] by Amiel Feinstein, Reprint of the 1964 English
  translation [ MR0173945]},
      review={\MR{3467631}},
}

\bib{GJ:87}{book}{
      author={Glimm, J.},
      author={Jaffe, A.},
       title={Quantum physics},
     edition={Second},
   publisher={Springer-Verlag, New York},
        date={1987},
        ISBN={0-387-96476-2},
         url={https://doi.org/10.1007/978-1-4612-4728-9},
        note={A functional integral point of view},
      review={\MR{887102}},
}

\bib{GKOS:01}{book}{
      author={Grosser, M.},
      author={Kunzinger, M.},
      author={Oberguggenberger, M.},
      author={Steinbauer, R.},
       title={Geometric theory of generalized functions},
    subtitle={with applications to relativity},
   publisher={Kluwer},
     address={Dordrecht},
        date={2001},
}

\bib{Grundling:88}{article}{
      author={Grundling, H.},
       title={Systems with outer constraints. {G}upta-{B}leuler
  electromagnetism as an algebraic field theory},
        date={1988},
        ISSN={0010-3616},
     journal={Comm. Math. Phys.},
      volume={114},
      number={1},
       pages={69\ndash 91},
         url={http://projecteuclid.org/euclid.cmp/1104160489},
      review={\MR{925522}},
}

\bib{Grundling:06}{article}{
      author={Grundling, H.},
       title={Quantum constraints},
        date={2006},
        ISSN={0034-4877},
     journal={Rep. Math. Phys.},
      volume={57},
      number={1},
       pages={97\ndash 120},
         url={https://doi.org/10.1016/S0034-4877(06)80011-X},
      review={\MR{2209784}},
}

\bib{GH:88CMP}{article}{
      author={Grundling, H.},
      author={Hurst, C.~A.},
       title={The quantum theory of second class constraints: kinematics},
        date={1988},
        ISSN={0010-3616},
     journal={Comm. Math. Phys.},
      volume={119},
      number={1},
       pages={75\ndash 93},
         url={http://projecteuclid.org/euclid.cmp/1104162271},
      review={\MR{968481}},
}

\bib{GH:85}{article}{
      author={Grundling, H. B. G.~S.},
      author={Hurst, C.~A.},
       title={Algebraic quantization of systems with a gauge degeneracy},
        date={1985},
        ISSN={0010-3616},
     journal={Comm. Math. Phys.},
      volume={98},
      number={3},
       pages={369\ndash 390},
         url={http://projecteuclid.org/euclid.cmp/1103942445},
      review={\MR{788780}},
}

\bib{GH:88}{article}{
      author={Grundling, H. B. G.~S.},
      author={Hurst, C.~A.},
       title={A note on regular states and supplementary conditions},
        date={1988},
        ISSN={0377-9017},
     journal={Lett. Math. Phys.},
      volume={15},
      number={3},
       pages={205\ndash 212},
         url={https://doi.org/10.1007/BF00398589},
      review={\MR{948354}},
}

\bib{Hoermander:V1}{book}{
      author={H{\"o}rmander, L.},
       title={The analysis of linear partial differential operators},
     edition={Second},
   publisher={Springer-Verlag},
        date={1990},
      volume={I},
}

\bib{Hoermann:91}{unpublished}{
      author={H{\"o}rmann, G.},
       title={The {W}eyl-algebra on the two-torus},
        date={1991},
        note={NTZ/INTSEM-Nr.34/1991 Preprint of the workshop ``Entropie und
  dynamische Entropie in der Mathematischen Physik'', University of Leipzig,
  August 1991},
}

\bib{Hoermann:93}{unpublished}{
      author={H{\"o}rmann, G.},
       title={Representations of the infinite dimensional {H}eisenberg group},
        date={1993},
        note={Doctoral thesis, University of Vienna},
}

\bib{Hoermann:97}{article}{
      author={H{\"o}rmann, G.},
       title={Regular {W}eyl-systems and smooth structures on {H}eisenberg
  groups},
        date={1997},
     journal={Commun. Math. Phys.},
      volume={184},
       pages={51\ndash 63},
}

\bib{KR:V2}{book}{
      author={Kadison, R.~V.},
      author={Ringrose, J.~R.},
       title={Fundamentals of the theory of operator algebras, volume {II}:
  Advanced theory},
   publisher={Academic Press},
     address={New York},
        date={1986},
}

\bib{Kirillov:76}{book}{
      author={Kirillov, A.~A.},
       title={Elements of the theory of representations},
   publisher={Springer-Verlag, Berlin-New York},
        date={1976},
        note={Translated from the Russian by Edwin Hewitt, Grundlehren der
  Mathematischen Wissenschaften, Band 220},
      review={\MR{0412321}},
}

\bib{KM:97}{book}{
      author={Kriegl, A.},
      author={Michor, P.~W.},
       title={The convenient setting of global analysis},
      series={Mathematical Surveys and Monographs},
   publisher={American Mathematical Society, Providence, RI},
        date={1997},
      volume={53},
        ISBN={0-8218-0780-3},
         url={https://doi.org/10.1090/surv/053},
      review={\MR{1471480}},
}

\bib{Pedersen:18}{book}{
      author={Pedersen, G.~K.},
       title={{$C^*$}-algebras and their automorphism groups},
      series={Pure and Applied Mathematics (Amsterdam)},
   publisher={Academic Press, London},
        date={2018},
        ISBN={978-0-12-814122-9},
        note={Second edition of [ MR0548006], Edited and with a preface by S\o
  ren Eilers and Dorte Olesen},
      review={\MR{3839621}},
}

\bib{Petz:90}{book}{
      author={Petz, D.},
       title={An invitation to the algebra of canonical commutation relations},
      series={Leuven Notes in Mathematical and Theoretical Physics. Series A:
  Mathematical Physics},
   publisher={Leuven University Press, Leuven},
        date={1990},
      volume={2},
        ISBN={90-6186-360-0},
      review={\MR{1057180}},
}

\bib{Slawny:72}{article}{
      author={Slawny, J.},
       title={On factor representations and the {$C^{\ast} $}-algebra of
  canonical commutation relations},
        date={1972},
        ISSN={0010-3616},
     journal={Comm. Math. Phys.},
      volume={24},
       pages={151\ndash 170},
         url={http://projecteuclid.org/euclid.cmp/1103857742},
      review={\MR{293942}},
}

\bib{Smith:52}{article}{
      author={Smith, M.~F.},
       title={The {P}ontrjagin duality theorem in linear spaces},
        date={1952},
        ISSN={0003-486X},
     journal={Ann. of Math. (2)},
      volume={56},
       pages={248\ndash 253},
         url={https://doi.org/10.2307/1969798},
      review={\MR{49479}},
}

\bib{NT:92}{incollection}{
      author={Thirring, W.},
      author={Narnhofer, N.},
       title={Covariant {QED} without indefinite metric},
        date={1992},
       pages={197\ndash 211},
         url={https://doi.org/10.1142/S0129055X92000200},
        note={Special issue dedicated to R. Haag on the occasion of his 70th
  birthday},
      review={\MR{1199174}},
}

\bib{Treves:06}{book}{
      author={Tr{\`e}ves, F.},
       title={Topological vector spaces, distributions and kernels},
   publisher={Dover Publications, Inc., Mineola, NY},
        date={2006},
        ISBN={0-486-45352-9},
        note={Unabridged republication of the 1967 original},
      review={\MR{2296978}},
}

\end{biblist}
\end{bibdiv}
\bibliographystyle{abbrv}

\end{document}